\documentclass[submission,copyright,creativecommons]{eptcs}
\providecommand{\event}{TARK 2023} % Name of the event you are submitting to

\usepackage{iftex}
\usepackage{amsmath,amsthm,amssymb,tikz}
\usepackage[english]{babel}
\usepackage{pifont}
\usepackage{wrapfig}
\usepackage{soul} %for highlights

\usetikzlibrary{decorations.pathmorphing}
\usetikzlibrary{shapes.geometric}

\usepackage{todonotes}
\newcommand\louwe[1]{\todo[color=green!30,size=\small,inline]{Louwe: #1}}
\newcommand\rustam[1]{\todo[color=blue!30,size=\small,inline]{Rustam: #1}}

\theoremstyle{definition}
\newtheorem{definition}{Definition}[section]
\newtheorem{lemma}{Lemma}
\newtheorem{theorem}{Theorem}

\newtheorem{corollary}{Corollary}
\newtheorem{remark}{Remark}

\ifpdf
  \usepackage{underscore}         % Only needed if you use pdflatex.
  \usepackage[T1]{fontenc}        % Recommended with pdflatex
\else
  \usepackage{breakurl}           % Not needed if you use pdflatex only.
\fi

\title{Satisfiability of Arbitrary Public Announcement Logic with Common Knowledge is $\Sigma^1_1$-hard}
\author{Rustam Galimullin
\institute{University of Bergen\\ Bergen, Norway}
\email{rustam.galimullin@uib.no}
\and
Louwe B. Kuijer
\institute{University of Liverpool\\
Liverpool, UK}
\email{lbkuijer@liverpool.ac.uk}
}
\def\titlerunning{Satisfiability of APALC is $\Sigma^1_1$-hard}
\def\authorrunning{R. Galimullin \& L. B. Kuijer}

\begin{document}
\maketitle

\begin{abstract}
Arbitrary Public Announcement Logic with Common Knowledge (APALC) is an extension of Public Announcement Logic with common knowledge modality and quantifiers over announcements. We show that the satisfiability problem of APALC on $S5$-models, as well as that of two other related logics with quantification and common knowledge, is $\Sigma^1_1$-hard. This implies that neither the validities nor the satisfiable formulas of APALC are recursively enumerable. Which, in turn, implies that APALC is not finitely axiomatisable.
\end{abstract}

\section{Introduction}
%The idea here would be to show that APALC, GALC and CALC are all unaxiomatizable. In fact, I think we can show that they are $\Sigma_1^1$-complete.

%The idea would be to use pretty much the same reduction as was used by Tim French et al.\ in the APAL undecidability paper. But instead of using a universal relation to make faraway worlds satisfy the same conditions as closeby ones, we use common knowledge.

%The advantage is that this allows use to make sure that we remain ``within the same grid'', which lets us use the stronger conditions necessary for the $\Sigma_1^1$-complete recurrent tiling problem.

\textbf{Quantified Public Announcement Logics}. 
\textit{Epistemic logic} (EL) \cite{meyer95} is one of the better-known formalisms for reasoning about knowledge of agents in multi-agent systems. It extends the language of propositional logic with constructs $\square_a \varphi$ meaning that `agent $a$ knows $\varphi$'. Formulas of EL are interpreted on epistemic models (or, equivalently, $S5$-models) that comprise a set of states, equivalence relations for each agent between states, and a valuation function that specifies in which states propositional variables are true. 
However, EL provides only a static description of distribution of knowledge in a system. Extensions of the logic that  allow one to reason about how information of individual agents and groups thereof changes as a result of some epistemic event are generally collectively known as \textit{dynamic epistemic logics} (DELs) \cite{del}.

%While EL provides a static description of distribution of knowledge in a system, there is a plenty of logics that allow one to reason about how information of indivudual agents and groups thereof changes as a result of some epistemic event. These logics are generally collectively known as \textit{dynamic epistemic logics} (DELs) \cite{del}. 
The prime example of a DEL and arguably the most well-studied logic in the family is \textit{public announcement logic} (PAL) \cite{plaza89}. A public announcement is an event of all agents publicly and simultaneously receiving the same piece of information. %Examples of such type of an event include showing your card to everyone in a parlour game, sending a message over a channel with a known eavesdropper, and, not surprisingly, making a public announcement.
The language of PAL extends that of EL with formulas $[\psi]\varphi$ that are read as `after public announcement of $\psi$, $\varphi$ is true'. %Semantically, an announcement of $\psi$ removes from a given model all states that do not satisfy the formula, and $\varphi$ is evaluated in the resulting model. 
%\louwe{I don't think the eavesdropper example is necessarily a public announcement. (That would require common knowledge of the eavesdropper listening in.)}

Quantification over various epistemic actions, and in particular over public announcements, has been explored in the last 15 or so years \cite{vanditmarsch23}. Adding quantification over public announcements allows one to shift the emphasis from the effects of a particular announcement to the question of (non-)existence of an announcement leading to a desired epistemic goal. In this paper, we focus on the three, perhaps most well-known, \textit{quantified PALs} (QPALs). The first of the three is \textit{arbitrary PAL} (APAL) \cite{balbiani08} that extends the language of PAL with constructs $[!]\varphi$ meaning `after \textit{any} public announcement, $\varphi$ is true'. A formula with the dual existential quantifier $\langle ! \rangle \varphi$ is read as `\textit{there is} a public announcement, after which $\varphi$ is true'. 

Observe that quantifiers of APAL do not specify whether an announcement can be made by any of the agents, or groups thereof, modelled in a system. Hence, a more `agent-centric' quantified PAL was proposed. \textit{Group announcement logic} (GAL) \cite{agotnes10} extends the language of PAL with formulas $[G]\varphi$ meaning `after \textit{any} announcement by agents from group $G$, $\varphi$ is true'. A formula with the dual of the universal GAL quantifier is $\langle G \rangle \varphi$ that is read `\textit{there is} an announcement by agents from group $G$ that makes $\varphi$ true'. %By an `announcement by agents from group $G$' we mean a truthful simultaneous announcement made by every agent in $G$. Particularly, the truthful part means that an agent in question actually knows the formula she wants to announce, i.e. $\square_a \psi_a$ for agent $a$. The simultaneous part means that a group announcement is a single announcement of the conjunction $\bigwedge_{a \in G} \square_a \psi_a$.

Once we start reasoning about what groups of agents can achieve by making public announcements, it is only too natural to consider their abilities in a game-theoretic setting. In particular, we may let agents outside of the group make their own announcements in an attempt to preclude the group from reaching their epistemic goals. A QPAL with such a competitive flavour to it is called \textit{coalition announcement logic} (CAL) \cite{agotnes08,thesis}. The logic extends PAL with modalities $[\!\langle G \rangle \!] \varphi$ that are read as `\textit{whatever} agents from coalition $G$ announce, \textit{there is} a counter-announcement by the anti-coalition that makes $\varphi$ true'. The diamond version $\langle \! [ G ] \! \rangle \varphi$ is then means that `\textit{there is} an announcement by coalition $G$, such that \textit{whatever} the anti-coalition announces at the same time, they cannot avoid $\varphi$'. 
Observe, that compared to APAL and GAL, modalities of CAL contain double quantification: $\forall \exists$ and $\exists \forall$ correspondingly. As the name of the logic suggests, modalities of CAL were inspired by coalition logic \cite{pauly02}, and they capture game-theoretic notions of $\alpha$- and $\beta$-effectivity \cite{aumann61}.

\textbf{Some Logical Properties of QPALs}. One of the most pressing open problems in the area is the existence of finitary axiomatisations of QPALs. Both finitary and infinitary axiom systems for APAL were proposed in \cite{balbiani08}, but later the finitary version was shown to be unsound \cite{kuijer15a}. The infinitary axiomatisation is, however, sound and complete \cite{balbiani15}. As the axiomatisation of GAL \cite{agotnes10} is quite similar to that of APAL, its finitary version is also not sound \cite[Footnote 4]{fan16}, and its infinitary version can be shown to be sound and complete by a modification of the proof from \cite{balbiani15}. To the best of our knowledge, there are no known sound and complete proof systems, finitary or infinitary, for CAL\footnote{A complete infinitary axiomatisation with CAL modalities and additional operators was given in \cite{galimullin21}}.

The satisfiability problem for QPALs is known to be undecidable \cite{agotnes16}. The result is achieved by a reduction from the classic tiling problem that consists in answering the question whether a given finite set of tiles can tile the $\mathbb{N} \times \mathbb{N}$ plane. Since this problem is co-RE-complete \cite{berger66,harel85}, or, equivalently, $\Pi^0_1$-complete, the reduction amounts to the fact that the satisfiability problem for QPALs is co-RE-hard (or $\Pi^0_1$-hard). Note that this result does not rule out the existence of finitary axiomatisations of QPALs. %Indeed, such an existence would follow if we could show that the satisfiability problem for QPALs is co-RE-complete. 
A prime example of a logic with a co-RE-complete satisfiability problem and a finitary axiomatisation is first-order logic.

\textbf{Overview of the paper and our result.} In this paper we consider extensions of QPALs with \textit{common knowledge} \cite{fagin95}, which is a classic variant of group knowledge in multi-agent systems. Its intuitive meaning is that `$\varphi$ is common knowledge among agents in group $G$ if everyone in $G$ knows $\varphi$, everyone in $G$ knows that everyone in $G$ knows $\varphi$ and so on ad infinitum'. Semantically, common knowledge among agents from $G$ corresponds to the reflexive transitive closure of equivalence relations of all agents from group $G$. We call extensions of APAL, GAL, and CAL with common knowledge APALC \cite{agotnes23}, GALC, and CALC, correspondingly, or QPALCs if we refer to all of them at the same time. 

%Adding common knowledge modalities $\blacksquare_G \varphi$ to the languages of QPALCs is not entirely straightforward. For APAL, GAL, and CAL, we assume that quantification ranges over a quantifier-free fragment of the language, i.e. over PAL. Restricting the range of quantification to EL does not have any effect as PAL is equally expressive as EL \cite{plaza89}. This is, however, not the case once we consider EL and PAL both extended with common knowledge (ELC and PALC correspondingly). PALC is strictly more expressive than ELC \cite{del}, and ELC, in its turn, is strictly more expressive than EL, and thus it matters, expressivity-wise, which quantifer-free fragment of a QPALC the quantification ranges over. This matters are explored in \cite{agotnes23}, where also infinitary axiomatisations of APALC and GALC are given. For our current purposes, though, the difference in the range of quantification does not play a role.

The result we prove in this paper is that the satisfiability problems for QPALCs are $\Sigma_1^1$-hard. We do this by showing that the \emph{recurring tiling problem}, which is known to be $\Sigma_1^1$-complete \cite{harel86}, can be reduced to satisfiability of QPALC formulas.
Because the satisfiability problems are $\Sigma_1^1$-hard, it follows that, in particular, the set of valid QPALC formulas is not recursively enumerable. That, in turn, implies that QPALCs have no finitary axiomatisations. The non-existence of a finitary axiomatisation of a somewhat related arbitrary arrow update logic \cite{vanditmarsch17} with common knowledge was shown in \cite{kuijer17} by the reduction from the non-halting problem. Moreover, the recurring tiling problem was used in \cite{miller05} to demonstrate that the satisfiability problem of  PAL with iterated announcements and common knowledge is $\Sigma^1_1$-complete.  
% The result we are showing in this paper is that for QPALCs neither validities nor satisfiable formulas are recursively enumerable. To achieve this, we present a reduction from the recurring tiling problem to the satisfiability problem of QPALCs. The recurring tiling problem is a tiling problem with an additional requirement that a special tile appears infinitely often in the first column of the $\mathbb{N} \times \mathbb{N}$ grid. It is known that this problem is $\Sigma^1_1$-complete \cite{harel86}, i.e. even harder than the standard tiling. Thus, reduction from the recurring tiling problem implies that the satisfiability of QPALCs is $\Sigma_1^1$-hard, meaning that QPALCs are neither RE nor co-RE, and hence do not admit of finitary axiomatisations. The non-existence of a finitary axiomatisation of a somewhat related arbitrary arrow update logic \cite{vanditmarsch17} with common knowledge was shown in \cite{kuijer17} by the reduction from the halting problem.

The use of common knowledge is instrumental in our paper, since it allows us to have a `tighter' grid than the ones from \cite{agotnes16} and \cite{french08}.  We deem our result important in at least two ways. First, the non-existence of finitary axiomatisations of QPALCs is interesting in its own right as it demonstrates that presence of common knowledge in QPALCs is a sufficient condition for $\Sigma^1_1$-hardness. Second, having both our construction (with common knowledge) and the constructions from \cite{agotnes16} and \cite{french08} side by side, allows one to flesh out crucial differences between $\Sigma^1_1$-hardness and $\Sigma_1^0$%$\Pi_1^0$
-hardness arguments, and, hopefully, move closer to tackling the open problem of (non-)existence of finitary axiomatisations of QPALs.

\textbf{Outline of the paper.} The rest of the paper is organised as follows. In Section \ref{sec:background} we cover the background on QPALCs. After that, in Section \ref{sec:mainpart}, we prove the main claim of this paper, and, finally, we conclude in Section \ref{sec:concl}.
%\louwe{I think it would be worthwhile to emphasize here that $\Sigma_1^1$-completeness is a far stronger conclusion than merely non-axiomatizability.}
%\rustam{Could you do that? I lack the intuition regarding the $\Sigma_1^1$-completeness. Also, a related question: how can one claim that a logic is not finitely axiomatisable if the logic is decidable? (e.g. Presburger arithmetic)}

\section{Quantified Public Announcement Logics with Common Knowledge}
\label{sec:background}
Let $A$ be a finite set of agents, and $P$ be a countable set of propositional variables. 

\begin{definition}
The \emph{languages of arbitrary public announcement logic with common knowledge} $\mathsf{APALC}$, \emph{group announcement logic with common knowledge} $\mathsf{GALC}$, and \emph{coalition announcement logic with common knowledge} $\mathsf{CALC}$ are inductively defined as 
\begin{alignat*}{3}
    &\mathsf{APALC} &&\thinspace \ni && \enspace \varphi ::= p \mid \neg \varphi \mid (\varphi \land \varphi) \mid \square_a \varphi \mid [\varphi]\varphi \mid \blacksquare_G \varphi \mid [!]\varphi\\
    &\mathsf{GALC} &&\thinspace \ni && \enspace \varphi ::= p \mid \neg \varphi \mid (\varphi \land \varphi) \mid \square_a \varphi \mid [\varphi]\varphi \mid \blacksquare_G \varphi \mid [G]\varphi\\
    &\mathsf{CALC} &&\thinspace \ni && \enspace \varphi ::= p \mid \neg \varphi \mid (\varphi \land \varphi) \mid \square_a \varphi \mid [\varphi]\varphi \mid \blacksquare_G \varphi \mid [\!\langle G \rangle \!]\varphi
\end{alignat*}
where $p \in P$, $a \in A$,  and $G \subseteq A$. Duals are defined as $\Diamond_a \varphi := \lnot \square_a \lnot \varphi$, $\langle \psi \rangle \varphi := \lnot [\psi]\lnot\varphi$, $\blacklozenge_G \varphi := \lnot \blacksquare_G \lnot \varphi$, $\langle ! \rangle \varphi := \lnot [!] \lnot \varphi$, $\langle G \rangle \varphi := \lnot [G] \lnot \varphi$ and $\langle \! [ G ] \! \rangle \varphi := \lnot [\! \langle G \rangle \! ] \lnot \varphi$. 

The fragment of $\mathsf{APALC}$ without $[!]\varphi$ is called \emph{public announcement logic with common knowledge} $\mathsf{PALC}$; the latter without $[\varphi]\varphi$ is \emph{epistemic logic with common knowledge} $\mathsf{ELC}$; $\mathsf{PALC}$ and $\mathsf{ELC}$ minus $\blacksquare_G \varphi$ are, correspondingly, \emph{public announcement logic}  $\mathsf{PAL}$ and \emph{epistemic logic} $\mathsf{EL}$. Finally, fragments of $\mathsf{APALC}$, $\mathsf{GALC}$ and $\mathsf{CALC}$ without $\blacksquare_G \varphi$ are called \emph{arbitrary public announcement logic} $\mathsf{APAL}$, \emph{group announcement logic} $\mathsf{GAL}$ and \emph{coalition announcement logic} $\mathsf{CAL}$ respectively.
\end{definition}

%`Everyone in group $G$ knows $\varphi$' is denoted by $\square_G \varphi := \bigwedge_{i \in G} \square_i \varphi$, and $\square^n_G \varphi$ is defined inductively as $\square^0_G \varphi := \varphi$ and $\square^{n+1}_G \varphi:= \square_G \square^n_G \varphi$ for all natural numbers $n$.  Expression $\Diamond^n_G \varphi$ is defined similarly by substituting diamonds instead of boxes.

\begin{definition}
A \emph{model} $M$ is a tuple $(S, \sim, V)$, where $S$ is a non-empty set of states, $\sim: A \rightarrow 2^{S \times S}$ gives an equivalence relation for each agent, and $V:P \rightarrow 2^S$ is the valuation function. By $\sim_G$ we mean reflexive transitive closure of $\bigcup_{a \in G} \sim_a$.
We will denote model $M$ with a distinguished state $s$ as $M_s$. %Whenever necessary, we refer to the elements of the tuple as $S_M$, $\sim_M$, and $V_M$.

% We call model $N$ a \emph{submodel} of $M$ if $S_N \subseteq S_M$, and $\sim_N$ and $V_N$ are restrictions of $\sim_M$ and $V_M$ to $S_N$. We will also write $M^X_s = (S^X, R^X,V^X)$, where $X \subseteq S$, $s\in X$, $S^X = X$, $\sim^X(a) = \sim_a \cap (X \times X)$ for all $a \in A$, and $V^X(p) = V(p) \cap X$ for all $p \in P$.
\end{definition}

We would like to stress that agent relations in our models are \textit{equivalence relations} (and hence our models are $S5$ models). The results of this paper do not generalise to arbitrary agent relations in any obvious way.

%\louwe{We should explicitly note that we are working on S5 models. Our results do not generalize to K in any obvious way.}

It is assumed that for group announcements, agents know the formulas they announce. In the following, we write $\mathsf{PALC}^G = \{\bigwedge_{i \in G} \square_i \psi_i \mid \textrm{for all } i \in G, \psi_i \in \mathsf{PALC}\}$ to denote the set of all possible announcements by agents from group $G$. We will use $\psi_G$ to denote arbitrary elements of $\mathsf{PALC}^G$.

\begin{definition} 
\label{def:semantics}
Let $M_s = (S, R, V)$ be a model, $p \in P$, $G \subseteq A$, and $\varphi, \psi \in \mathsf{APALC} \cup \mathsf{GALC} \cup \mathsf{CALC}$.
\begin{alignat*}{3}
	&M_s \models p &&\quad \text{iff} &&\quad s \in V(p)\\
	&M_s \models \lnot \varphi &&\quad \text{iff} &&\quad M_s \not \models \varphi\\
	&M_s \models \varphi \land \psi &&\quad \text{iff} &&\quad M_s \models \varphi \text{ and } M_s \models \psi\\
	&M_s \models \square_a \varphi &&\quad \text{iff} &&\quad \forall t \in S: s \sim_a t \text{ implies } M_t \models \varphi \\
    &M_s \models \blacksquare_G \varphi  &&\quad \text{iff} &&\quad \forall t \in S: s \sim_G t \text{ implies } M_t \models \varphi\\
    &M_s \models [\psi] \varphi &&\quad \text{iff} &&\quad M_s \models \psi \text{ implies } M_s^\psi \models \varphi\\
     &M_s \models [!] \varphi  &&\quad \text{iff} &&\quad \forall \psi \in \mathsf{PALC}: M_s \models [\psi] \varphi\\
     &M_s \models [G] \varphi  &&\quad \text{iff} &&\quad \forall \psi_G \in \mathsf{PALC}^G: M_s \models [\psi_G] \varphi\\
     &M_s \models [\! \langle G\rangle \!] \varphi  &&\quad \text{iff} &&\quad \forall \psi_G \in \mathsf{PALC}^G, \exists \chi_{A \setminus G} \in \mathsf{PALC}^{A\setminus G}: M_s \models \psi_G \text{ implies } M_s \models \langle \psi_G \land \chi_{A \setminus G}\rangle \varphi
\end{alignat*}
where $M_s^\psi = (S^\psi, R^\psi, V^\psi)$ with $S^\psi = \{s \in S \mid M_s \models \psi\}$, $R^\psi (a)$ 
%= R(a) \cap S^\psi$ 
is the restriction of $R(a)$ to $S^\psi$
for all $a \in A$, and $V^\psi(p) = V(p) \cap S^\psi$ for all $p \in P$. 
\end{definition} 

Observe, that it follows from the definition of the semantics that in the case of the grand coalition $A$, $M_s \models [A] \varphi$ if and only if $M_s \models [\!\langle A \rangle \! ]\varphi$. For the case of the empty group $\varnothing$, we assume that the conjunction of an empty set of formulas is a tautology.

\begin{remark}
 For APAL, GAL, and CAL, we assume that quantification ranges over a quantifier-free fragment of the language, i.e. over PAL, which is equally expressive as EL \cite{plaza89}. 
 %Restricting the range of quantification to EL does not have any effect as PAL is equally expressive as EL \cite{plaza89}. 
 This is, however, not as straightforward once we consider ELC and PALC. The latter is strictly more expressive than ELC \cite[Theorem 8.48]{del}, and ELC, in its turn, is strictly more expressive than EL, and thus it matters, expressivity-wise, which quantifer-free fragment of a QPALC the quantification ranges over. These matters are explored in \cite{agotnes23}, where also infinitary axiomatisations of APALC and GALC are given. For our current purposes, though, the difference in the range of quantification does not play a role.
\end{remark}

\section{The Satisfiability Problem of QPALCs is $\Sigma^1_1$-hard}
\label{sec:mainpart}
We prove the $\Sigma^1_1$-hardness of the satisfiability problem of QPALCs via a reduction from the recurring tiling problem \cite{harel85}. 

\begin{definition}
    Let $C$ be a finite set of \emph{colours}. A \emph{tile} is a function $\tau:\{\mathsf{north}, \mathsf{south}, \mathsf{east}, \mathsf{west}\} \to C$. A finite set of tiles $\mathrm{T}$ is called an \emph{instance} of the tiling problem. A \emph{solution} to an instance of the tiling problem is a function\footnote{Throughout the paper we assume that $0 \in \mathbb{N}$.} $f:\mathbb{N} \times \mathbb{N} \to \mathrm{T}$ such that for all $(i,j) \in \mathbb{N} \times \mathbb{N}$,
    \begin{align*}
    f(i,j) (\mathsf{north}) = f(i,j+1) (\mathsf{south}) \text{ and } 
    f(i,j) (\mathsf{east}) = f(i+1,j) (\mathsf{west}). 
\end{align*}
\end{definition}

\begin{definition}
    Let $\mathrm{T}$ be a finite set of tiles with a designated tile $\tau^\ast \in \mathrm{T}$. The \emph{recurring tiling problem} is the problem to determine whether there is a solution to instance $\mathrm{T}$ of the tiling problem such that $\tau^\ast$ appears \textit{infinitely} often in the first column. 
\end{definition}
We assume without loss of generality that the designated tile $\tau^\ast$ occurs only in the first column.

\subsection{Encoding a Tiling}

For our construction we will require five propositional variables --- \textsf{north}, \textsf{south}, \textsf{east}, \textsf{west} and \textsf{centre} --- to designate the corresponding sides of tiles. Additionally, we will have designated propositional variables for each colour in $C$, and for each tile $\tau_i \in \mathrm{T}$ there is a propositional variable $p_i$ that represents this tile. Finally, we will use $p^\ast$ for the special $\tau^\ast$. %Finally, we will create a ``checkerboard'' pattern, and use propositional variables \textit{black} and \textit{white} for this purpose.

In our construction, we will represent each tile with (at least) five states: one for each of the four sides of a tile, and one for the centre. As for agents, we require only three of them for our construction. Agent $s$, for \textit{s}quare, cannot distinguish states within the same tile. Agent $v$, for \textit{v}ertical, cannot distinguish between the northern part of one tile and the southern part of the tile above. Similarly, the \textit{h}orizontal agent $h$ cannot distinguish between the eastern and western parts of adjacent tiles. See Figure \ref{fig:tiling} for the depiction of an intended grid-like model.

\begin{figure}[h!]
\centering
\begin{tikzpicture}
\fill (-1,-1) circle (0.1) node[below]{\footnotesize{\{$\mathsf{west}, c_2\}$}};
\fill (0.5, -1) circle (0.1) node[below]{\footnotesize{\{$\mathsf{centre}\}$}};
\fill (0.5, 0.5) circle (0.1) node[below]{\footnotesize{\{$\mathsf{north}, c_1\}$}};
\fill (2,-1) circle (0.1) node[below]{\footnotesize{\{$\mathsf{east}, c_4\}$}};
\fill (0.5,-2.5) circle (0.1) node[above]{\footnotesize{\{$\mathsf{south}, c_3\}$}};
\draw[dashed, rounded corners] (-1.75,-2.75) rectangle (2.75,0.75);
\fill (0.5, 1.5) circle (0.1) node[above]{\footnotesize{\{$\mathsf{south}, c_1\}$}};
\draw[thick] (0.5, 0.5) -- node[right]{$v$} (0.5, 1.5);
\fill (0.5, -3.5) circle (0.1) node[below]{\footnotesize{\{$\mathsf{north}, c_3\}$}};
\draw[thick] (0.5, -2.5) -- node[left]{$v$} (0.5, -3.5);
\fill (-3, -1) circle (0.1) node[below]{\footnotesize{\{$\mathsf{east}, c_2\}$}};
\draw[thick] (-1,-1) -- node[above]{$h$} (-3, -1);
\fill (4,-1) circle (0.1) node[below]{\footnotesize{\{$\mathsf{west}, c_4\}$}};
\draw[thick] (2,-1) -- node[below]{$h$} (4,-1);
\node (tile) at (-1.25, 0.25) {\Large{$\tau_i$}};

\draw[dashed, rounded corners] (-1.75, 2) -- (-1.75,1.25) -- (2.75,1.25) -- (2.75, 2);
\draw[dashed, rounded corners] (-3, 0.75) -- (-2.25, 0.75) -- (-2.25, -2.75) -- (-3, -2.75);
\draw[dashed, rounded corners] (4, 0.75) -- (3.25, 0.75) -- (3.25, -2.75) -- (4, -2.75);
\draw[dashed, rounded corners] (-1.75, -4) -- (-1.75,-3.25) -- (2.75,-3.25) -- (2.75, -4);
\end{tikzpicture}
\hspace{1cm}
\begin{tikzpicture}
\node[draw, minimum size = 20pt] (tile00) at (0,1) {$\tau_i$};
\node[draw, minimum size = 20pt] (tile01) at (0,2.5) {$\tau_j$};
\node[draw, minimum size = 20pt] (tile02) at (0,4) {$\tau_k$};
\node (tile03) at (0,5) {};
\draw[thick] (tile00) -- node[left]{$v$} (tile01);
\draw[thick] (tile01) -- node[left]{$v$} (tile02);
\draw[thick] (tile02) -- node[left]{$v$} (tile03);

\node[draw, minimum size = 20pt] (tile10) at (1.5,1) {$\tau_k$};
\node[draw, minimum size = 20pt] (tile20) at (3,1) {$\tau_j$};
\node (tile30) at (4,1) {};
\draw[thick] (tile00) -- node[above]{$h$} (tile10);
\draw[thick] (tile10) -- node[above]{$h$} (tile20);
\draw[thick] (tile20) -- node[above]{$h$} (tile30);

\node[draw, minimum size = 20pt] (tile11) at (1.5,2.5) {$\tau_i$};
\node[draw, minimum size = 20pt] (tile21) at (3,2.5) {$\tau_k$};
\node (tile31) at (4,2.5) {};
\draw[thick] (tile11) -- node[above]{$h$} (tile01);
\draw[thick] (tile21) -- node[above]{$h$} (tile11);
\draw[thick] (tile21) -- node[above]{$h$} (tile31);
\draw[thick] (tile11) -- node[left]{$v$} (tile10);
\draw[thick] (tile21) -- node[left]{$v$} (tile20);

\node[draw, minimum size = 20pt] (tile12) at (1.5,4) {$\tau_j$};
\node[draw, minimum size = 20pt] (tile22) at (3,4) {$\tau_i$};
\node (tile33) at (4,4) {};
\draw[thick] (tile02) -- node[above]{$h$} (tile12);
\draw[thick] (tile22) -- node[above]{$h$} (tile12);
\draw[thick] (tile22) -- node[above]{$h$} (tile33);
\draw[thick] (tile12) -- node[left]{$v$} (tile11);
\draw[thick] (tile22) -- node[left]{$v$} (tile21);

\node (tile13) at (1.5,5) {};
\node (tile23) at (3,5) {};
\draw[thick] (tile13) -- node[left]{$v$} (tile12);
\draw[thick] (tile23) -- node[left]{$v$} (tile22);

\node (shadow) at (0.5, -0.5) {};
\end{tikzpicture}
\caption{Left: a representation of a single tile $\tau_i$, where agent $s$ has the universal relation within the dashed square, relations $h$ and $v$ are equivalences, and reflexive arrows are omitted. Each state is labelled by a set of propositional variables that are true there. Right: an example of a grid-like model that we construct in our proof. Each tile $\tau$ has a similar structure as presented on the left of the figure.}
\label{fig:tiling}
\end{figure}
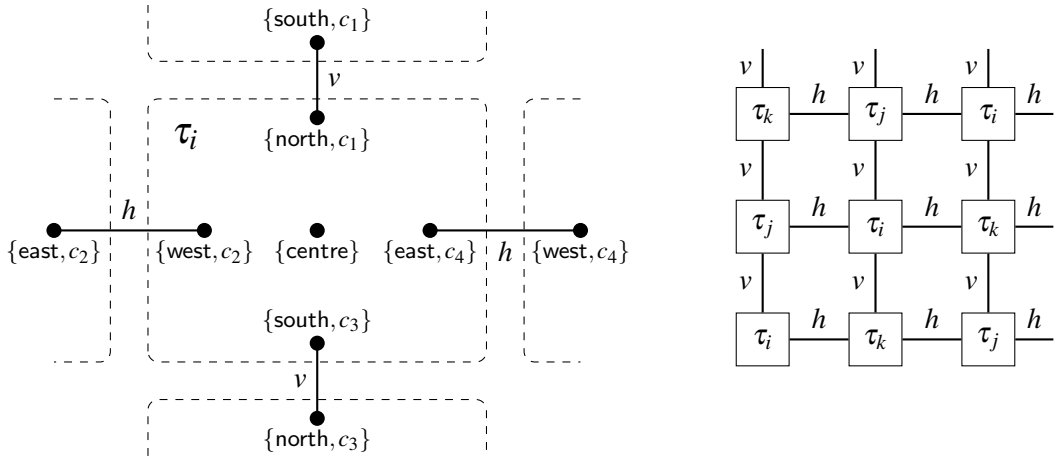

Let an instance $\mathrm{T}$ of the recurring tiling problem be given. We start by construction of formula $\Psi_{\mathrm{T}}$ that will be satisfied in a given model if and only if the model is grid-like.
We will build up $\Psi_{\mathrm{T}}$ step-by-step, defining useful subformulas along the way. Let $\mathsf{Position}$ be the following set $\mathsf{Position} := \{\mathsf{north},$ $\mathsf{south},$ $\mathsf{east},$ $\mathsf{west}, \mathsf{centre}\}$.

The first constraint, expressed by formula $\mathit{one\_colour}$, is that each state is coloured by exactly one colour. To ensure that all five parts --- north, south, east, west, and centre --- are present in a current square, we state in $\mathit{all\_parts}$ that in all squares the square agent $s$ has access to all five  relevant states.
\begin{align*}
    \mathit{one\_colour} := &\bigvee_{c \in C} \left(c \land \bigwedge_{d \in C \setminus \{c\}} \lnot d\right) &\mathit{all\_parts} := &\square_s \bigvee_{{q \in \mathsf{Position}}} q \land \bigwedge_{q \in \mathsf{Position}} \Diamond_s q
\end{align*}
The formulas $\mathit{hor}$ and $\mathit{vert}$ state that the relation $h$ only allows us to move between $\mathsf{east}$ and $\mathsf{west}$ states, while $v$ only allows movement between $\mathsf{north}$ and $\mathsf{south}$ states.
\begin{align*}
    \mathit{hor} := & {} \bigwedge_{q\in \{\mathsf{north},\mathsf{south},\mathsf{centre}\}} (q\rightarrow \square_h q) &
    \mathit{vert} := & {} \bigwedge_{q\in \{\mathsf{east},\mathsf{west},\mathsf{centre}\}} (q\rightarrow \square_v q)
\end{align*}
With $\mathit{one\_pos}$ we force each state to satisfy exactly one propositional variable from $\mathsf{Position}$, and with $\mathit{one\_tile}$ we ensure that all states within the same tile are labelled by the tile proposition.
\begin{align*}
    \mathit{one\_pos} := &\bigvee_{q \in \mathsf{Position}} \left(q \land \bigwedge_{q^\prime \in \mathsf{Position}\setminus\{q\}} \lnot q^\prime\right)
    & \mathit{one\_tile} := &\bigvee_{\tau_i \in \mathrm{T}} \left(p_i \land  \square_{s} p_i \land \bigwedge_{\tau_j \in \mathrm{T} \setminus \{\tau_i\}} \lnot p_j\right)
\end{align*}
Next, we force each state in a tile to satisfy exactly one atom corresponding to their designated colour:
\[\mathit{state\_col} := \bigvee_{\tau_i \in \mathrm{T}} \left(p_i \to \bigwedge_{q \in \mathsf{Position}\setminus \{\mathsf{centre}\}} (q \to \tau_i(q))\right), \]
where $\tau_i(q)$ is the colour of the tile $\tau_i$ on the side $q$ (e.g. $\tau_i (\mathsf{south})$ is the bottom colour of tile $\tau_i$).

All the formulas considered so far deal with the representation of a single tile. We will use the following abbreviation:
\[\psi_{\mathit{tile}} := \mathit{one\_colour} \land \mathit{all\_parts} \land \mathit{hor} \land \mathit{vert} \land \mathit{one\_pos} \land \mathit{one\_tile} \land \mathit{state\_col}\]

Adjoining tiles are required to have the same colour on the sides facing each other, we simulate this by requiring that agents $h$ and $v$ consider a current colour in the top and right directions. In such a way we also ensure that the grid is infinite in the positive quadrant.
\[\mathit{adj\_tiles} :=  \bigwedge_{c \in C} \left(
     (\mathsf{north} \land c \to \Diamond_v  \mathsf{south}  \land \square_v c) \land
      (\mathsf{east} \land c \to \Diamond_h \mathsf{west}  \land \square_h c)
    \right)
 \]
We are concerned with the reduction from the $\mathbb{N} \times \mathbb{N}$ recurring tiling problem, i.e. our grid will have left and bottom edges. We force the existence of a tile at position $(0,0)$ with the following formula:
\begin{align*}\mathit{init} := & {} \blacklozenge_{\{h,v,s\}} (\blacksquare_{\{v,s\}}(\mathsf{west}\rightarrow \square_h\mathsf{west}) \land \blacksquare_{\{h,s\}}(\mathsf{south}\rightarrow \square_v\mathsf{south}))
% \square_s ((\mathsf{south} \to \square_v \mathsf{south}) \land (\mathsf{west} \to \square_h \mathsf{west})) \land \\& ((\mathsf{west}\wedge\square_h \mathsf{west}) \rightarrow \blacksquare_{v,c}(\mathsf{west}\rightarrow \square_h\mathsf{west})) \land\\
% &((\mathsf{south}\wedge\square_v \mathsf{south}) \rightarrow \blacksquare_{h,c}(\mathsf{south}\rightarrow \square_v\mathsf{south}))
\end{align*}

For the remaining formulas, it is useful to define two abbreviations. We use $\square_\mathit{up}\varphi$ to denote $\square_s (\mathsf{north}\rightarrow \square_v(\mathsf{south}\rightarrow \varphi))$, i.e., we first move, by agent $s$, to the state representing the northern quadrant of the tile, then we move, by agent $v$, to southern quadrant of the tile above, where we evaluate $\varphi$. Similarly, we use $\square_\mathit{right}\varphi$ to denote $\square_s(\mathsf{east}\rightarrow\square_h(\mathsf{west}\rightarrow \varphi))$. The duals $\lozenge_\mathit{up}$ and $\lozenge_\mathit{right}$ are defined as usual.

% We first make sure that the checkerboard pattern is satisfied.
% \begin{align*}
% \mathit{checkerboard} := & ((\mathit{black}\wedge \neg \mathit{white})\vee (\neg\mathit{black}\wedge \mathit{white}))\wedge (\mathit{black}\leftrightarrow \square_s\mathit{black}) \wedge \\
% & (\mathit{black}\rightarrow (\square_\mathit{up}\mathit{white}\wedge \square_\mathit{right}\mathit{white})\wedge (\mathit{white}\rightarrow (\square_\mathit{up}\mathit{black}\wedge \square_\mathit{right}\mathit{black})
% \end{align*}

The next two formulas are used to guarantee that for every tile there are unique tiles, up to PALC-indistinguishability, above it and to its right.
\begin{align*}
    \mathit{up} := & [!](\lozenge_\mathit{up}\lozenge_s\mathsf{centre}\rightarrow\square_\mathit{up}\lozenge_s\mathsf{centre})\\
    \mathit{right} := & [!](\lozenge_\mathit{right}\lozenge_s\mathsf{centre}\rightarrow\square_\mathit{right}\lozenge_s\mathsf{centre})
\end{align*}
Additionally, we use the following two formulas to establish a commutative property: going \textit{up} and then \textit{right} results in a state that is PALC-indistinguishable from going \textit{right} and then \textit{up}.

\begin{align*}
    \mathit{right\&up} := [!](\lozenge_\mathit{right}\lozenge_\mathit{up}\lozenge_s\mathsf{centre}\rightarrow \square_\mathit{up}\square_\mathit{right}\lozenge_s\mathsf{centre})\\
    \mathit{up\&right} := [!](\lozenge_\mathit{up}\lozenge_\mathit{right}\lozenge_s\mathsf{centre}\rightarrow \square_\mathit{right}\square_\mathit{up}\lozenge_s\mathsf{centre})
\end{align*}

Finally, we make sure that any two states that are $h$ or $v$ related and that are in the same position are parts of indistinguishable tiles.

\begin{align*}
    \mathit{no\_change} := & {} \bigwedge_{q,q'\in \mathsf{Position}}[!]((q\wedge\lozenge_s q')\rightarrow (\square_h (q\rightarrow \lozenge_s q')\wedge \square_v (q\rightarrow\lozenge_sq')))
\end{align*}
The formula $\mathit{hor}$ states that unless we are in a $\mathsf{east}$ or $\mathsf{west}$ position, we cannot go to a different position using $h$. Similarly, $\mathit{vert}$ states that unless we are in a $\mathsf{north}$ or $\mathsf{south}$ position we can't use $v$ to change position. The formula $\mathit{no\_change}$ then states that any move by relation $h$ or $v$ that does not change the position must lead to an indistinguishable tile.

% other than the necessary ones, e.g., we can't use $h$ to move to a different tile except from the $\mathsf{east}$ and $\mathsf{west}$ edges.

% from the wrong states. For example, from the $\mathsf{east}$ or $\mathsf{centre}$ part of a tile, the relation $v$ must not lead to a tile that is distinguishable from the current one.
% \begin{align*}
%     \mathit{no\_right} := \bigwedge_{\mathit{q}\in \{\mathsf{north},\mathsf{south},\mathit{center}\}}[!]((q\wedge \lozenge_s\mathsf{east})\rightarrow \square_v(q\wedge\lozenge_s\mathsf{east}))\\
%     \mathit{no\_up} := \bigwedge_{\mathit{q}\in \{\mathsf{east},\mathsf{west},\mathit{center}\}}[!]((q\wedge \lozenge_s\mathsf{north})\rightarrow \square_v(q\wedge\lozenge_s\mathsf{north}))
% \end{align*}

We abbreviate formulas with quantifiers as 
\[\psi_{x\&y}:= \mathit{up} \land \mathit{right} \land \mathit{right\&up} \land \mathit{up\&right} \land \mathit{no\_change}\]
In our reduction, we are interested in grids where a special tile appears infinitely often in the first column of the grid. %We designate such a tile with $\tau^\ast$ and assume that a given instance of the tiling problem includes $\tau^\ast$. A corresponding propositional variable for the tile is $p^\ast$. 
The following formula requires that the special tile appears only in the leftmost column:
\[\mathit{tile\_left} := p^\ast \to \square_s(\mathsf{west} \to \square_{h} \mathsf{west}) \]

All of this completes the necessary requirements for the grid. Now, by adding a common knowledge modality for all agents, we force all of the aforementioned formulas to hold everywhere in the grid.
\[\Psi_{\mathrm{T}} := \blacksquare_{\{h,v,s\}} \left( 
     \psi_{\mathit{tile}} \land
      \mathit{adj\_tiles}\land
      \mathit{init}\land
      \psi_{\mathit{x\&y}}\land
      \mathit{tile\_left}
\right)
\]
Observe that $\Psi_{\mathrm{T}}$ does not say anything about the special tile $\tau^\ast$ appearing infinitely often in the first column. The formula merely requires that if there is a special tile, then it should appear in the first column. We first show that $\Psi_\mathrm{T}$ forces a grid-like model, and only after that will we consider the (in)finite number of occurrences of the special tile. 

\begin{lemma}
\label{lem:there}
Let $\mathrm{T}$ be an instance of the recurring tiling problem. If $\mathrm{T}$ can tile $\mathbb{N} \times \mathbb{N}$, then $\Psi_{\mathrm{T}}$ is satisfiable.
\end{lemma}

\begin{proof}
Assume that there is a tiling of the $\mathbb{N} \times \mathbb{N}$ plane with a finite set of tiles $\mathrm{T}$. We construct model $M = (S, \sim, V)$ satisfying $\Psi_{\mathrm{T}}$ directly from the given tiling. In particular, 
\begin{itemize}
    \item $S = \mathbb{N} \times \mathbb{N} \times \{\mathfrak{n}, \mathfrak{s}, \mathfrak{e}, \mathfrak{w}, \mathfrak{c}\}$, 
    \item $\sim_s = \{(i,j,\mathfrak{l}), (i',j',\mathfrak{l}') \mid i = i' \text{ and } j = j'\}$
    \item $\sim_v$ is the reflexive closure of $\{(i,j,\mathfrak{n}), (i, j+1, \mathfrak{s})\}$
    \item $\sim_h$ is the reflexive closure of $\{(i,j,\mathfrak{e}), (i+1, j, \mathfrak{w})\}$
    \item for all $\tau_k \in \mathrm{T}$, $V(p_k) = \{(i,j,\mathfrak{l}) \mid \tau_k \text{ is at } (i,j)\}$
    \item for all $c \in C$, $V(c) = \{(i,j,\mathfrak{l}) \mid \tau(\mathfrak{l})=c\}$
    \item for all $l \in \mathsf{Position}$, $V(l) = \{(i,j,\mathfrak{l}) \mid l \text{ corresponds to } \mathfrak{l}\}$
\end{itemize}

To argue that $M_{(0,0,\mathfrak{e})} \models \Psi_{\mathrm{T}}$ we first notice that due to the fact that $\mathrm{T}$ tiles the $\mathbb{N} \times \mathbb{N}$ plane and by the construction of $M$, subformulas of $ \Psi_{\mathrm{T}}$ that do not involve arbitrary announcements are straightforwardly satisfied. 

Now, consider the formula $\mathit{up}$. For every $(i,j,\mathfrak{l})$, there is at most one $(i',j',\mathfrak{l}')$ that is reachable by taking an $s$-step to a $\mathsf{north}$ state followed by a $v$-step to a $\mathsf{south}$ state, namely $(i',j',\mathfrak{l}')=(i,j+1,\mathfrak{s})$. Furthermore, this property is retained in any submodel of $M$. As a consequence, in any state of any submodel of $M$,  $\lozenge_\mathit{up}\chi$ implies $\square_\mathit{up}\chi$, for every $\chi$. In particular, it follows that $M_{(i,j,\mathfrak{l})}\models [!](\lozenge_\mathit{up}\lozenge_s\mathsf{centre}\rightarrow \square_\mathit{up}\lozenge_s\mathsf{centre})$, i.e., $M_{(i,j,\mathfrak{l})}\models \mathit{up}$.

Similar reasoning shows that $(i,j,\mathfrak{l})$ satisfies the other conjuncts of $\psi_{x\&y}$. Hence $M_{(i,j,\mathfrak{l})}\models \psi_{\mathit{tile}} \land
      \mathit{adj\_tiles}\land
      \mathit{init}\land
      \psi_{\mathit{x\&y}}\land
      \mathit{tile\_left}$, for all $(i,j,\mathfrak{l})$, and thus $M_{(0,0,\mathfrak{e})}\models\Psi_\mathrm{T}$.
\end{proof}
The more complex part of the reduction is to show that if $\Psi_\mathrm{T}$ is satisfiable, then a tiling exists.
\begin{lemma}
\label{lem:and_back}
    Let $\mathrm{T}$ be an instance of the recurring tiling problem. If $\Psi_{\mathrm{T}}$ is satisfiable, then $\mathrm{T}$ can tile $\mathbb{N} \times \mathbb{N}$.
\end{lemma}
\begin{proof}
    Let $M$ be such that $M_s\models \Psi_\mathrm{T}$. The model $M$ is partitioned by $\sim_s$, we refer to these partitions as grid points, and label these points as follows.
    \begin{itemize}
        \item The grid point containing $s$ is labelled $(0,0)$.
        \item If $A$ and $B$ are grid points, $A$ is labelled $(i,j)$ and there is a $\mathsf{north}$-state in $A$ that is $v$-indistinguishable to a $\mathsf{south}$-state in $B$, then $B$ is labelled $(i,j+1)$.
        \item If $A$ and $B$ are grid points, $A$ is labelled $(i,j)$ and there is a $\mathsf{east}$-state in $A$ that is $h$-indistinguishable to a $\mathsf{west}$-state in $B$, then $B$ is labelled $(i+1,j)$.
    \end{itemize}
    Note that a single grid point might have multiple labels.
    We say that $(i,j)$ is tiled with $\tau_i$ if there is some grid point labelled with $(i,j)$ that contains a state where $p_i$ holds. We start by noting that because the main connective of $\Psi_\mathrm{T}$ is $\blacksquare_{\{h,v,s\}}$, the formula holds in every labelled grid point. For every labelled grid point $X$ and every $x\in X$, we therefore %, in particular, 
    have $M_x\models \psi_\mathit{tile}$. So $X$ contains states for every direction, each labelled with exactly one colour that corresponds to the tile that holds on $X$. We continue by proving the following claim.

    \vspace{10pt}
    \textbf{Claim 1:} Let $X$, $A$ and $B$ be grid points where $X$ is labeled $(i,j)$ while $A$ and $B$ are both labeled $(i,j+k)$ by virtue of being $k$-steps to the north of $X$. Then $A$ and $B$ are PALC-indistinguishable, in the sense that for every $\chi\in \mathsf{PALC}$, if there is an $a\in A$ such that $M_a\models \chi$ then there is a $b\in B$ such $M_b\models \chi$ (and vice versa).

    \textbf{Proof of Claim 1:} By induction on $k$. As base case, let $k=1$ and suppose towards a contradiction that, for some $\chi\in \mathsf{PALC}$ and $a \in A$, $M_a\models \chi$ while for every $b\in B$, $M_b\not\models \chi$. Consider then the formula $\mathsf{centre}\rightarrow \lozenge_s\chi$. Every $\mathsf{centre}$ state in $A$ satisfies this formula, while none of the $\mathsf{centre}$ states in $B$ do. Hence, for every state $x\in X$, $M_x\models [\mathsf{centre}\rightarrow\lozenge_s\chi](\lozenge_\mathit{up}\lozenge_s\mathsf{centre}\land \neg \square_\mathit{up}\lozenge_s\mathsf{centre})$. But that contradicts $M_x\models \mathit{up}$. From this contradiction, we prove the base case $k=1$.

    Now, suppose as induction hypothesis that $k>1$ and that the claim holds for all $k'<k$. Again, suppose towards a contradiction that $M_a\models \chi$ while $M_b\not\models \chi$ for all $b\in B$. Let $A'$ and $B'$ be grid points that lie $k-1$ steps to the north of $X$ and one step to the south of $A$ and $B$, respectively. Then for every $a'\in A'$ and $b'\in B'$, $M_{a'}\models \lozenge_\mathit{up}\lozenge_s\chi$ and $M_{b'}\models \lozenge_\mathit{up}\neg\lozenge_s\chi$. By the induction hypothesis, $A'$ and $B'$ are indistinguishable, so $M_{a'}\models \lozenge_\mathit{up}\lozenge_s\chi\wedge \lozenge_\mathit{up}\neg\lozenge_s\chi$. But then there are distinguishable grid points one step to the north of $A'$, contradicting the induction hypothesis. From this contradiction, we prove the induction step and thereby the claim.
    
    \vspace{10pt}
    Similar reasoning shows that any two grid points $A, B$ that are labeled $(i+k,j)$ by virtue of being $k$ steps to the right of the same grid point $X$ are indistinguishable. Now, we can prove the next claim.

    \vspace{10pt}
    \textbf{Claim 2:} Let $X$, $A$ and $B$ be grid points, where $X$ is labelled $(i,j)$, $A$ is labelled $(i+1,j+1)$ by virtue of being above $A'$ which is to the right of $X$, and $B$ is labelled $(i+1,j+1)$ by virtue of being to the right of $B'$ which is above $B$. Then $A$ and $B$ are PALC-indistinguishable.

    \textbf{Proof of claim 2:} Suppose towards a contradiction that for some $\chi\in \mathsf{PALC}$ and $a \in A$ we have $M_a\models \chi$, while $M_b\not\models \chi$ for all $b\in B$. Then for $x\in X$ we have $M_x\models [\mathsf{centre}\rightarrow \lozenge_s\chi](\lozenge_\mathit{right}\lozenge_\mathit{up}\lozenge_s\mathsf{centre} \wedge \lozenge_\mathit{up}\lozenge_\mathit{right}\neg\lozenge_s\mathsf{centre})$, contradicting $M_x\models \mathit{right\&up}$.
    
    \vspace{10pt}
    From Claim 1 it follows that any $A$ and $B$ that are labelled $(i,j)$ by virtue of being $i$ steps to the right and then $j$ steps up from $(0,0)$ are PALC-indistinguishable. Claim 2 then lets us commute the ``up'' and ``right'' moves. Any path to $(i,j)$ can be obtained from the path that first goes right $i$ steps then up $j$ steps by a finite sequence of such commutations. Hence any grid points $A$ and $B$ that are labelled $(i,j)$ are PALC-indistinguishable.

    The tile formulas $p_i$, for every $\tau_i\in \mathrm{T}$, are PALC-formulas, so there is exactly one tile $\tau_i$ that is assigned to the grid point $(i,j)$.  
    Furthermore, $\mathit{state\_col}$ then guarantees that each side of a grid point has the colour corresponding to the tile, and $\mathit{adj\_tiles}$ guaranteees that the tile colours match. This shows that if $\Psi_\mathrm{T}$ is satisfiable, then $\mathrm{T}$ can tile $\mathbb{N}\times \mathbb{N}$.
    \end{proof}

\subsection{Encoding the Recurring Tile}
The final formula that is satisfied in a grid model if and only if a given tiling has a tile that occurs infinitely often in the first column would be \[\Psi_{\mathrm{T}} \land \blacksquare_{\{v,s\}}[\blacksquare_{\{h,s\}} \lnot p^\ast] \lnot \Psi_{\mathrm{T}}.\] 
In other words, the recurring tiling problem can be reduced to the APALC-satisfiability problem, where the reduction maps the instance $(\mathrm{T},\tau^\ast)$ of the recurring tiling problem to the satisfiability of $\Psi_{\mathrm{T}} \land \blacksquare_{\{v,s\}}[\blacksquare_{\{h,s\}} \lnot p^\ast] \lnot \Psi_{\mathrm{T}}$.

Intuitively, the formula states that if we remove all rows with the special tile, then our model is no longer a grid. See Figure \ref{fig:infingrid}, where on the left we have a grid with the special grey tile $\tau^\ast$ appearing infinitely often in the first column (every other tile in the first column is grey). Formula $\blacksquare_{\{h,s\}} \lnot p^\ast$ holds only in those squares of the grid that lie on rows without the special tile. Thus, announcing $\blacksquare_{\{h,s\}} \lnot p^\ast$ removes all rows that has the grey tile (see the right part of Figure \ref{fig:infingrid}). Since the grey tile appears infinitely often in the original grid, we have to remove an infinite number of rows after the announcement of $\blacksquare_{\{h,s\}} \lnot p^\ast$, thus ensuring that what is left of the original model is not a grid.

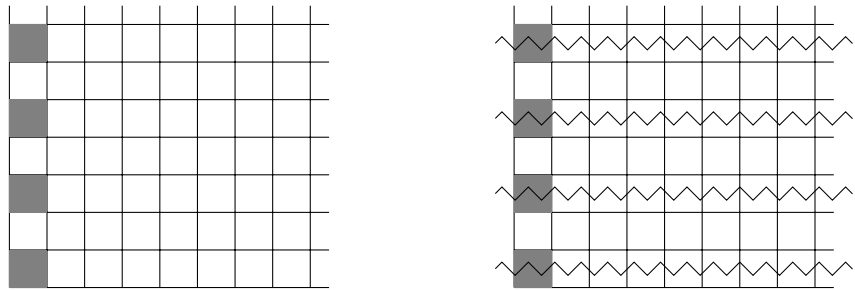
\begin{figure}[h!]
    \centering
    \begin{tikzpicture}
    \draw[step=0.5,black,thin] (0,0) grid (4.25,3.75);
    \node[fill=gray, minimum size=0.5cm ] at (0.25,0.25){};
    \node[fill=gray, minimum size=0.5cm ] at (0.25,1.25){};
    \node[fill=gray, minimum size=0.5cm ] at (0.25,2.25){};
    \node[fill=gray, minimum size=0.5cm ] at (0.25,3.25){};
\end{tikzpicture}
\hspace{2cm}
\begin{tikzpicture}
    \draw[step=0.5,black,thin] (0,0) grid (4.25,3.75);
    \node[fill=gray, minimum size=0.5cm ] at (0.25,0.25){};
    \node[fill=gray, minimum size=0.5cm ] at (0.25,1.25){};
    \node[fill=gray, minimum size=0.5cm ] at (0.25,2.25){};
    \node[fill=gray, minimum size=0.5cm ] at (0.25,3.25){};
    \draw[decorate,decoration={zigzag}] (-.25, .25) -- (4.5, .25);
    \draw[decorate,decoration={zigzag}] (-.25, 1.25) -- (4.5, 1.25);
    \draw[decorate,decoration={zigzag}] (-.25, 2.25) -- (4.5, 2.25);
    \draw[decorate,decoration={zigzag}] (-.25, 3.25) -- (4.5, 3.25);
\end{tikzpicture}
    \caption{Left: An original grid with a special grey tile $\tau^\ast$ appearing infinitely often in the first column. Right: The grid after the public announcement of $\blacksquare_{\{h,s\}}\lnot p^\ast$. Crossed-out rows are not preserved after the announcement. }
    \label{fig:infingrid}
\end{figure}

\begin{theorem}
     Let $\mathrm{T}$ be an instance of the tiling problem with a special tile $\tau^\ast \in \mathrm{T}$. Set $\mathrm{T}$ can tile $\mathbb{N} \times \mathbb{N}$ with $\tau^\ast$ appearing infinitely often in the first column if and only if $\Psi_{\mathrm{T}} \land \blacksquare_{\{v,s\}}[\blacksquare_{\{h,s\}} \lnot p^\ast] \lnot \Psi_{\mathrm{T}}$ is satisfiable.
\end{theorem}

\begin{proof}
    First, let us can extend the labelling from the proof of Lemma~\ref{lem:and_back} as follows:
    \begin{itemize}
        \item For every $q\in \mathsf{Position}$, if $A$ and $B$ are grid points, $A$ is labeled $(i,j)$ and there is a $q$ state in $A$ that is $v$ or $h$-indistinguishable from a $q$ state in $B$, then $B$ is labeled $(i,j)$.
    \end{itemize}
    
    It follows from $\mathit{no\_change}$ that this extended labelling retains the property that any two grid points with the same label are PALC-indistinguishable. Furthermore, from $\mathit{hor}$ and $\mathit{vert}$ it follows that every grid point that is reachable by $h$, $v$ and $s$ is now labelled with some coordinates $(i,j)$. Hence we can identify the $\{h, v, s\}$-reachable grid points in any model of $\Psi_\mathrm{T}$ with $\mathbb{N}\times \mathbb{N}$.

    Now, assume that set $\mathrm{T}$ cannot tile the $\mathbb{N} \times \mathbb{N}$ plane with a special tile $\tau^\prime \in \mathrm{T}$ appearing infinitely often in the first column. 
    We argue that in this case, $\Psi_{\mathrm{T}} \land \blacksquare_{\{v,s\}}([\blacksquare_{\{h,s\}} \lnot p^\ast] \lnot \Psi_{\mathrm{T}})$ is not satisfiable. 
    The first conjunct is straightforward. If $\mathrm{T}$ cannot tile the $\mathbb{N} \times \mathbb{N}$ plane, then, by Lemma \ref{lem:and_back}, $\Psi_{\mathrm{T}}$ is not satisfiable. 

    So suppose that $\mathrm{T}$ can tile the plane, but only in such a way that $\tau^\ast$ occurs finitely often. For every model $M_{(0,0,\mathfrak{l})}$ of $\Psi_\mathrm{T}$, there is then some $k\in \mathbb{N}$ that is the last row in which $p^\ast$ is true. The formula $\blacksquare_{\{h,s\}}\neg p^\ast$ holds exactly on those rows where $p^\ast$ does not hold in the first column. As a result, the update $[\blacksquare_{\{h,s\}}\neg p^\ast]$ does not remove any rows past row $k$. The grid points $\mathbb{N}\times \mathbb{N}_{>k}$ then still form a grid that is isomorphic to $\mathbb{N}\times \mathbb{N}$, and that is tiled. See Figure~\ref{fig:fingrid} for a depiction of the situation.

    It follows that %$M_{(0,k,\mathfrak{l})}^{\blacksquare_{\{h,s\}}\neg p^\ast}\models \Psi_\matrhm{T}$, so 
    $M_{(0,k,\mathfrak{l})}\not\models [\blacksquare_{\{h,s\}}\neg p^\ast]\neg \Psi_\mathrm{T}$, and therefore $M_{(0,0,\mathfrak{l})}\not\models \blacksquare_{\{v,s\}}[\blacksquare_{\{h,s\}}\neg p^\ast]\neg\Psi_\mathrm{T}$.
    This is true for every model of $\Psi_\mathrm{T}$, so $\Psi_\mathrm{T}\land \blacksquare_{\{v,s\}}[\blacksquare_{\{h,s\}}\neg p^\ast]\neg\Psi_\mathrm{T}$ is not satisfiable.

    If, on the other hand, $\mathrm{T}$ can tile the plane in such a way that $\tau^\ast$ occurs infinitely often in the first column, then there is a model of $\Psi_\mathrm{T}$ where the modality $[\blacksquare_{\{h,s\}}\neg p^\ast]$ removes infinitely many rows, and therefore does not leave any infinite grid. So $\Psi_\mathrm{T}\wedge \blacksquare_{\{v,s\}}[\blacksquare_{\{h,s\}}\neg p^\ast]\neg \Psi_\mathrm{T}$ is satisfiable.
    \end{proof}
    % If $\mathrm{T}$ can tile the $\mathbb{N} \times \mathbb{N}$ plane, but $\tau^\ast$ never appears in the first column, then $\blacklozenge_{\{v,s\}} p^\ast$ is trivially false. Finally, assume that $\mathrm{T}$ can tile the $\mathbb{N} \times \mathbb{N}$ plane, and $\tau^\ast$ appears in the first column only \textit{finitely} often. This means that there is some position $(0,j)$ in the grid, such that $\tau^\ast$ does not cover it and any other positions above (i.e. $(0,j)$ is above the final occurrence of $\tau^\ast$ in the first column). 
    % By the construction of $M$, we have that there is a state $(0,j,\mathfrak{l})$ with $(0,k,\mathfrak{l}) \not \in V(p^\ast)$ for all $k \geqslant j$. As  $(0,j,\mathfrak{l})$ is in the first column, it is reachable by $\sim_{\{v,s\}}$. As argued above, announcement of $\blacksquare_{\{h,s\}} \lnot p^\ast$ removes all rows in $M$ that do not start with a $p^\ast$-state. However, since the number of occurrences of $p^\ast$ is finite, and $(0,j,\mathfrak{l})$ is above the last row that satisfied $p^\ast$, $M_{(0,j,\mathfrak{l})}^{\blacksquare_{\{h,s\}} \lnot p^\ast} \models \Psi_{\mathrm{T}}$. See Figure \ref{fig:fingrid} for the depiction of the situation.

\begin{figure}[h!]
    \centering
    \begin{tikzpicture}
    \draw[step=0.5,black,thin] (0,0) grid (4.25,3.75);
    \node[fill=gray, minimum size=0.5cm ] at (0.25,0.25){};
    \node[fill=gray, minimum size=0.5cm ] at (0.25,1.25){};
    \node[fill=gray, minimum size=0.5cm ] at (0.25,1.75){};
\end{tikzpicture}
\hspace{2cm}
\begin{tikzpicture}
    \draw[step=0.5,black,thin] (0,0) grid (4.25,3.75);
    \node[fill=gray, minimum size=0.5cm ] at (0.25,0.25){};
    \node[fill=gray, minimum size=0.5cm ] at (0.25,1.25){};
    \node[fill=gray, minimum size=0.5cm ] at (0.25,1.75){};
    \draw[decorate,decoration={zigzag}] (-.25, .25) -- (4.5, .25);
    \draw[decorate,decoration={zigzag}] (-.25, 1.25) -- (4.5, 1.25);
    \draw[decorate,decoration={zigzag}] (-.25, 1.75) -- (4.5, 1.75);
    \draw[step=0.5,black, very thick] (0,1.99) grid (4.25,3.75);
\end{tikzpicture}
    \caption{Left: An original grid with a special grey tile $
    \tau^\ast$ appearing finitely often in the first column. Right: The grid after the public announcement of $\blacksquare_{\{h,s\}}\lnot p^\ast$. Crossed-out rows are not preserved after the announcement. A full $\mathbb{N} \times \mathbb{N}$ grid that is still available after the announcement is depicted with thick lines.}
    \label{fig:fingrid}
\end{figure}
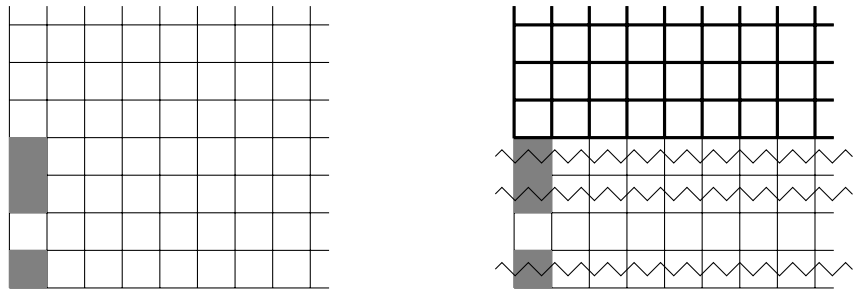

In the construction of $\Psi_{\mathrm{T}}$ and proofs of Lemmas \ref{lem:there} and \ref{lem:and_back}, we used APALC quantifiers $[!]$. We can prove the similar results for GALC and CALC quantifers by substituting $[!]$ with $[\{h,v,s\}]$ and $[ \! \langle \{h,v,s\} \rangle \! ]$ correspondingly, and substituting $\mathsf{PALC}$ with $\mathsf{PALC}^{\{h,v,s\}}$.  We get the hardness result from the $\Sigma^1_1$-completeness of the recurring tiling problem \cite{harel86}.

\begin{theorem}
The satisfiability problem of QPALCs is $\Sigma^1_1$-hard.
\end{theorem}

The $\Sigma^1_1$-hardness of the satisfiability problems of QPALCs together with the fact that the class of $\Sigma^1_1$ problems is strictly greater than the class of co-RE problems \cite[Chapter 4]{odifreddi89} imply that the sets of validites of the logics are not RE, which, in turn, implies that QPALCs are not finitely axiomatisable.

\begin{corollary}
The set of valid formulas of QPALCs is neither RE nor co-RE.
\end{corollary}

\begin{corollary}
\label{cor:nonfin}
    QPALCs do not have finitary axiomatisations.
\end{corollary}

\section{Discussion}
\label{sec:concl}
The existence of finitary axiomatisations of any of APAL, GAL, and CAL is a long-standing open problem. In this paper, we have showed that the satisfiability problem of the logics extended with common knowledge modality is $\Sigma^1_1$-hard, and thus they do not admit of finitary axiomatisations. Table \ref{table:results} contains the overview of the known results, including those shown in this paper, and open questions.
\begin{table}[h!]
\centering
\begin{tabular}{l*{6}{c}}
              & APAL & GAL & CAL & APALC & GALC  & CALC \\
\hline
Finitary axiomatisation & ? & ? & ? & \ding{55} (Cor. \ref{cor:nonfin}) & \ding{55} (Cor. \ref{cor:nonfin}) & \ding{55} (Cor. \ref{cor:nonfin})   \\
Infinitary axiomatisation& \ding{52}\cite{balbiani08} & \ding{52}\cite{agotnes10} & ? & \ding{52}\cite{agotnes23} &  \ding{52}\cite{agotnes23} & ?  \\
\end{tabular}
\caption{Overview of the known results and open problems.}
\label{table:results}
\end{table}

It is important to point out that the use of common knowledge is instrumental in our construction. Arguments from \cite{french08,agotnes16} did not rely on common knowledge to enforce local grid properties globally, and instead the authors used an agent with the universal relation over the set of states. This approach is good enough if one wants to demonstrate the existence of a grid-like model. However, if we also require that the grid satisfies some property, like a special tile occurring infinitely often in the first column, then the presence of the global agent makes it harder to ensure this. The problem is that such an unrestrained relation may access other grids within the same model, and thus we 
%cannot guarantee that is a grid that satisfies the desired property. We 
may end up in the situation when the property is satisfied by a set of grids taken together and not by any single grid. 

Our construction is `tighter' than those in \cite{french08,agotnes16}. In particular, our \textit{v}ertical and \textit{h}orizontal agents can `see' only one step ahead. This guarantees that we stay within the same grid. In order to force grid properties globally, we use common knowledge operators that allow us to traverse a given grid-like model in all directions. It is not yet clear how to have a `tight' grid and still be able to traverse the model without common knowledge. %If we can solve this problem, then we will answer the open question of the existence of finitary axiomatisations of APAL, GAL, and CAL. 
With this work, apart from showing that QPALCs are $\Sigma^1_1$-hard,
%answering negatively the same question about APALC, GALC, and CALC, 
we also hope to have elucidated the exact obstacle one has to overcome in order to claim the same about QPALs. 

\subsection*{Acknowledgements}
We would like to thank the three anonymous reviewers for their encouraging comments and constructive suggestions, which helped us to improve the presentation of our result. 

\bibliographystyle{eptcs} 
\bibliography{theultimatebibliography.bib}
\end{document}